\documentclass[prodmode,acmtist]{acmsmall} 

\usepackage[ruled]{algorithm2e}
\usepackage{amsmath}

\SetAlFnt{\small}
\SetAlCapFnt{\small}
\SetAlCapNameFnt{\small}
\SetAlCapHSkip{0pt}
\IncMargin{-\parindent}

\newtheorem{fact}{Fact}

\begin{document}

\markboth{X. Wang et al.}{TensorBeat: Tensor Decomposition for Monitoring Multi-Person Breathing Beats with Commodity WiFi}

\title{TensorBeat: Tensor Decomposition for Monitoring Multi-Person Breathing Beats with Commodity WiFi}
\author{Xuyu Wang
\affil{Auburn University, Auburn, AL}
Chao Yang
\affil{Auburn University, Auburn, AL}
Shiwen Mao
\affil{Auburn University, Auburn, AL}}

\begin{abstract}
Breathing signal monitoring can provide important clues for human's physical health problems. Comparing to existing techniques that require wearable devices and special equipment, a more desirable approach is to provide contact-free and long-term breathing rate monitoring by exploiting wireless signals. In this paper, we propose TensorBeat, a system to employ channel state information (CSI) phase difference data to intelligently estimate breathing rates for multiple persons with commodity WiFi devices. The main idea is to leverage the tensor decomposition technique to handle the CSI phase difference data. The proposed TensorBeat scheme first obtains CSI phase difference data between pairs of antennas at the WiFi receiver to create CSI tensor data. Then Canonical Polyadic (CP) decomposition is applied to obtain the desired breathing signals. A stable signal matching algorithm is developed to find the decomposed signal pairs, and a peak detection method is applied to estimate the breathing rates for multiple persons. Our experimental study shows that TensorBeat can achieve high accuracy under different environments for multi-person breathing rate monitoring.
\end{abstract}
%
%


%
%


\terms{Health sensing, breathing rate estimation, tensor decomposition,
stable roommate matching, commodity WiFi, channel state information, phase difference}


\begin{bottomstuff}
This work is supported in part by the U.S. NSF under Grants CNS-1247955 and CNS-1320664, and by the Wireless Engineering Research and Education Center (WEREC) at Auburn University. 
Author's addresses: X. Wang, C. Yang, {and} S. Mao, 
Department of Electrical and Computer Engineering, Auburn University, Auburn, AL 36849-5201 USA. 
\end{bottomstuff}

\maketitle

\section{Introduction}

It is estimated that 100 million Americans suffer chronic health conditions such as lung disorders, diabetes, and heart disease~\cite{Call}. About three-fourths of the total US healthcare cost is spent on dealing with these health conditions. To reduce such costs, there is an increasing demand for long-term health monitoring in indoor environments. By tracking vital signs such as breathing and heart beats, the patient's physical health can be timely evaluated and meaningful clues for medical problems can be provided~\cite{H}. For example, monitoring breathing signals can help identify sleep disorders or anomalies for patients, as well as decreasing sudden infant death syndrome (SIDS) for sleeping infants~\cite{SIDS}. Traditional approaches for monitoring vital signs require patients to wear special devices, such as a capnometer~\cite{Capnometers} to estimate breathing rate, or a pulse oximeter~\cite{photoplethysmographic} on the finger to track heart beats. Recently, smartphones are used to estimate breathing rate by employing the built-in gyroscope, accelerometer~\cite{Zephyr}, and microphone~\cite{Sleep}, and for physical activity recognition using accelerometer~\cite{Tao2016}. 
This requires the patient to place a smartphone near-by and wear sensors in the monitoring environment. Moreover, the readily available smartphone sensors such as accelerometer and gyroscope can only monitor breathing rate for single person. Such existing approaches could be expensive, inconvenient to use, and annoying even for a short period of time. An alternative approach is in need to provide a contact-free and long-term breathing monitoring at low costs.

Recently, several RF based systems for vital signs tracking are proposed, which employ wireless signal to monitor breathing-induced chest movements and are mainly based on radar and WiFi techniques. For radar based vital signals monitoring, several techniques, such as the Doppler radar~\cite{Doppler} and ultra-wideband radar~\cite{ultrawideband}, are used to track vital signs, which require special hardware operated at high frequency and at a high cost. Moreover, the Vital-Radio system employs a frequency modulated continuous wave (FMCW) radar to track breathing and heart rates~\cite{smart}; it requires a customized hardware using a wide bandwidth from 5.46 GHz to 7.25 GHz. For WiFi based vital signs monitoring, mmVital~\cite{Millimeter} can exploit the received signal strength (RSS) of 60 GHz millimeter wave (mmWave) signals for breathing and heart rate estimation. mmVital also operates with a larger bandwidth of about 7 GHz, and uses a customized hardware with a mechanical rotator. Another technique UbiBreathe uses WiFi RSS for breathing rate estimation, which, however, requires the device be placed in the line of sight (LOS) path between the transmitter and the receiver~\cite{UbiBreathe}.

Compared with RSS, channel state information (CSI) provides fine-grained channel information in the physical layer, which can now be read from modified device drivers for several off-the-shelf WiFi network interface cards (NIC), such as the Intel WiFi Link 5300 NIC~\cite{Predictable} and the Atheros AR9580 chipset~\cite{Yaxiong}. Moreover, CSI includes both amplitude and phase values at the subcarrier-level for orthogonal frequency-division multiplexing (OFDM) channels, which is a more  stable and accuracy representation of channel characteristics than RSS, including the non-LOS (NLOS) components for small-scale fading. Recently, the authors in~\cite{Tracking} use the CSI amplitude data to monitor breathing and heart signals, which requires the person to remain in the sleeping mode. However, the measured CSI phase data has not been fully exploited in prior works, largely due to 
random phase fluctuation resulting from asynchronous times and frequencies of the transmitter and receiver.  
For multiple person breathing monitoring, because the reflected components in the received signal are from the chests of multiple persons, each moves slightly due to breathing and the movements are independent. Thus, vital signs monitoring and estimation for multiple persons still remains a challenging and open problem.

In this paper, we propose to utilize CSI phase difference data between antenna pairs to monitor the breathing rates of multiple persons. First, we show that when the person is in a stationary state, such as standing, sitting, or sleeping, the CSI phase difference data is 
highly stable in consecutively received packets, which can be leveraged for extracting the small, periodic breathing signal hidden in the received WiFi signal. In fact, phase difference is more robust than amplitude, which usually exhibits large fluctuations because of the attenuation over the link distance, obstacles, and the multipath effect. 
Moreover, the phase difference data captures and preserves the periodicity of breathing, when the wireless signal is reflected from the patients' chests. 
To extract the weak breathing signal, and more important, to distinguish among multiple persons, we propose to employ a tensor decomposition method to handle the phase difference data~\cite{tensor4,Luo2015,Hu2014}. 
We create the CSI tensor data by increasing the dimension of CSI data from one to three, which can be used to effectively separate different breathing signals in different clusters. 

We present a system termed {\em TensorBeat}, {\em Tensor} decomposition for estimating multiple persons breathing {\em Beats}, by exploiting CSI phase difference data. TensorBeat operates as follows. First, it obtains 60 CSI phase difference data from antenna pairs 1 and 2, and 2 and 3, at the receiver. 
Next, a data preprocessing procedure is applied to the measured phase difference data, including data calibration and Hankelization. In the data calibration phase, the direct current (DC) component and high frequency noises are removed. In the Hankelization phase, a two dimensional Hankel matrix is created based on the calibrated phase difference data from every subcarrier, and the rank of the Hankel matrix is analyzed. Then, we adopt 
Canonical Polyadic (CP) decomposition for estimating multiple persons breathing signs, and prove the uniqueness of the proposed CSI tensor. 
After CP decomposition, we obtain twice amount of breathing signals, which, however, are randomly indexed. We thus design a stable signal matching algorithm (for the stable roommate problem~\cite{matching1}) to identify the decomposed signal pairs for each person. 
Finally, we combine the decomposed signals in each pair and employ a peak detection method to estimate the breathing rate for each person.

We implement TensorBeat on commodity 5 GHz WiFi devices and verify its performance with five persons over six months in different indoor environments, such as a computer laboratory, a through-wall scenario, and a long corridor. The results show that the proposed TensorBeat system can achieve high accuracy and high success rates for multiple persons breathing rates estimation. Moreover, we demonstrate the robustness of the proposed TensorBeat system for monitoring multiple persons' breathing beats under a wide range of environmental parameters.

The main contributions of this paper are summarized as follows. 
\begin{enumerate}
	\item We theoretically and experimentally verify the feasibility of leveraging CSI phase difference for breathing monitoring. In particular, we analyze the measured phase errors in detail and demonstrate that phase difference data is stable and can be used to extract breathing signs. To the best of our knowledge, we are the first to leverage phase difference for multiple persons breathing rate estimation.  
	\item We are also the first to apply tensor decomposition for RF sensing based vital signs monitoring. We use the phase difference data to create a CSI tensor for all subcarrier at the three antennas of the WiFi receiver. We then incorporate CP decomposition to obtain the desired breathing signals. A stable signal matching algorithm is developed to match the decomposed signals for each person, while a peak detection method is used to estimate multiple persons' breathing rates.  
	\item We prototype the TensorBeat system with commodity 5 GHz WiFi devices and demonstrate its superior performance in different indoor environments with extensive experiments. The results show that the proposed TensorBeat system can achieve very high accuracy and high success rates for multiple persons breathing rate estimation.  
\end{enumerate}

The remainder of this paper is organized as follows. 
The preliminaries and phase difference analysis are provided in Section~\ref{sec:sysMod}. 
We present the TensorBeat system design and performance analysis in Section~\ref{sec:ad} and verify its performance with extensive experiments in Section~\ref{sec:sml}. We provide the related work in Section~\ref{sec:related}. Section~\ref{sec:conC} concludes this paper.

\section{Preliminaries and Phase Difference Information}\label{sec:sysMod}

\subsection{Tensor Decomposition Preliminaries}
A tensor is considered as a multidimensional array~\cite{tensor1}. The dimensions of the tensor are called as modes, and the order of the tensor is the number of the modes. For example, the $N$-order tensor is a $N$-mode tensor. Moreover, It is noticed that a first-order tensor is a vector, a second-order tensor is a matrix, and a third-order tensor is a cubic structure. Higher-order tensors with $(N \geq 3)$ have a wide range of applications such as data mining, brain data analysis, recommendation systems, wireless communications, computer vision, and healthcare and medical applications~\cite{tensor4}. For higher-order tensors, they face various computational challenging because of the exponential increase in time and space complexity with the orders increase of tensors. This leads to the curse of dimensionality. Fortunately, tensor decomposition as one powerful tool is leveraged for alleviating the curve by decomposing high-order tensors into a limited number of factors. Also, it can obtain hidden feature components, thus extracting physical insight of higher-order tensors. Two main tensor decompositions are tucker decomposition and CP decomposition~\cite{tensor1}. We consider CP decomposition for multiple persons breathing rate estimation because it can easily obtain the unique solution~\cite{tensor1}. On the other hand, we will provide some necessary definitions and equations of tensor decomposition, which can be used for our proposed algorithm.

\begin{definition} \label{D1}
 (Frobenius Norm of a Tensor). The Frobenius norm of a tensor $\chi \in \mathbb{K}^{I_1 \times I_2 \times \cdots \times I_N}$ is the square root of the sum of the squares of all its elements, which is defined by
\begin{eqnarray} \label{eq:D1}
  \left\| \chi \right\|_{F} =\sqrt{\sum_{i_1=1}^{I_1}\sum_{i_2=1}^{I_2} \cdots \sum_{i_N=1}^{I_N} x_{i_1,i_2 \cdots i_N}^2 }.
\end{eqnarray}
\end{definition}
where $\mathbb{K}$ stands for $\mathbb{R}$ or $\mathbb{C}$.

\begin{definition} \label{D2}
 (Kronecker Product). The Kronecker product of matrics $\bm A \in \mathbb{K}^{I \times J}$ and $\bm B \in \mathbb{K}^{M \times N}$ is denoted as $\bm A \otimes \bm B$. The result is an $(IM) \times (JN) $ matrix, which is defined by
\begin{eqnarray} \label{eq:D2}
  \bm A \otimes \bm B =\begin{bmatrix}
    a_{11}B & a_{12}B & \dots  & a_{1J}B \\
    a_{21}B & a_{22}B & \dots  & a_{2J}B \\
    \vdots & \vdots & \vdots & \vdots \\
    a_{I1}B & a_{I2}B & \dots  & a_{IJ}B
\end{bmatrix}.
\end{eqnarray}
\end{definition}

\begin{definition} \label{D3}
 (Khatri-Rao Product). The Khatri-Rao product of $A \in \mathbb{K}^{I \times J}$ and $B \in \mathbb{K}^{M \times J}$ is denoted as $A \odot B$. It is the column-wise Kronecker product with the size $(IM) \times J$, which is defined by
\begin{eqnarray} \label{eq:D3}
  \bm A \odot \bm B = [\bm a_1 \otimes \bm b_1, \bm a_2 \otimes \bm b_2, \cdots  ,\bm a_J \otimes \bm b_J]. 
\end{eqnarray}
\end{definition}

\begin{definition} \label{D4}
 (Hadamard product). The Hadamard product of $\bm A \in \mathbb{K}^{I \times J}$ and $\bm B \in \mathbb{K}^{I \times J}$ is denoted as $\bm A * \bm B$. It is the elementwise matrix product with the size $I \times J$, which is defined by
\begin{eqnarray} \label{eq:D4}
 \bm A * \bm B = \begin{bmatrix}
    a_{11}b_{11} & a_{12}b_{12} & \dots  & a_{1J}b_{1J} \\
    a_{21}b_{21} & a_{22}b_{22} & \dots  & a_{2J}b_{2J} \\
    \vdots & \vdots & \vdots & \vdots \\
    a_{I1}b_{I1} & a_{I2}b_{I2} & \dots  & a_{IJ}b_{IJ}
\end{bmatrix}. 
\end{eqnarray}
\end{definition}
\subsection{Channel State Information Preliminaries}

OFDM is an effective wireless transmission technique widely used in many wireless systems, including WiFi (such as IEEE 802.11 a/g/n) and LTE~\cite{Wang16,Xu14b}. 
The OFDM system partitions the wireless channel into multiple orthogonal subcarriers, where data is transmitted over all the subcarriers by using the same modulation and coding scheme (MCS) to combat frequency selective fading. With modified device driver for off-the-shelf NICs, such as the Intel 5300 NIC~\cite{Predictable} and the Atheros AR9580 chipset~\cite{Yaxiong}, the CSI data can be extracted, which represents fine-grained physical (PHY) information. Moreover, CSI captures rich wireless channel characteristics such as shadowing fading, distortion, and the multipath effect.

The WiFi OFDM channel can be regarded as a narrowband flat fading channel, which can be expressed in the frequency domain as 
\begin{eqnarray} \label{eq:CSI1}
  \vec{Y}=\mbox{H} \cdot \vec{X}+\vec{N}, 
\end{eqnarray}
where $\vec{Y}$ and $\vec{X}$ denote the received and transmitted wireless signal vectors, respectively, $\vec{N}$ is the additive white Gaussian noise, and $\mbox{H}$ represents the channel frequency response, which can be estimated from $\vec{Y}$ and $\vec{X}$. 

Although the WiFi OFDM system can use 56 subcarriers for data transmission on a 20 MHz channel, the Intel 5300 NIC device driver can only provide CSI for 30 out of the 56 subcarriers using the channel bonding technique. The channel frequency response of subcarrier $i$, denoted by $\mbox{H}_i$, is a complex value, given as
\begin{eqnarray} \label{eq:CSI 2}
  \mbox{H}_i=\mathcal{I}_i+j {\mathcal{Q}_i} =|\mbox{H}_i|\exp{\left( j {\angle{\mbox{H}_i}} \right)}, 	
\end{eqnarray}
where $\mathcal{I}_i$ and $\mathcal{Q}_i$ are the in-phase component and quadrature component, respectively; $|\mbox{H}_i|$ and $\angle{\mbox{H}_i}$ are the amplitude 
and phase response of subcarrier $i$, respectively.  

For indoor environments with multipath components, the channel frequency response of subcarrier $i$, $\mbox{H}_i$, can also be written as 
\begin{eqnarray} \label{eq:CSI 3}
  \mbox{H}_i=\sum_{k=0}^{K}r_k \cdot e^{-j2 \pi f_i \tau_k}, 
\end{eqnarray}
where $K$ is the number of multipath components, and $r_k$ and $\tau_k$ are the attenuation and propagation delay on the $k_{th}$ path, respectively.

Traditionally, the multipath components are regarded as harmful for wireless communications, since they cause the delay spread (requires guard times) and large fluctuation of received wireless signal (harder to demodulate). 
For indoor localization systems, 
multiple signals will be received from a single transmission, 
including one LOS signal and many reflected signals. It is a challenging problem to detect the LOS signal from the mixed multipath components, which is indicative of the direction of the transmitter~\cite{RSS_CSI,XWang14}. 
In this paper, however, we take a different view and show that the multiple signals reflected from the chests of multiple persons can be useful for estimating their breathing rates simultaneously. 

\subsection{Phase Difference Information}

As discussed, we exploit phase difference information for breathing rate estimation. We verify that the phase difference values between two adjacent antennas are stable for consecutively received packets in this section. In fact, the extracted phase information from the Intel 5300 NIC is high random and cannot be used for breathing monitoring. This is because of the asynchronous times and frequencies of transmitter and receiver NICs. Recently, two effective techniques are proposed for CSI phase calibration, to remove the unknown random components in CSI phase data. 
The first technique is to take a linear transformation for the CSI phase data over all the subcarirers~\cite{PADS,PhaseFi_J,PhaseFi}. 
The other technique is to use the phase difference 
between two adjacent antennas in the 2.4 GHz band, and to remove the measured average of phase difference for LOS recognition~\cite{PhaseU}. 
It can be seen that these techniques only obtain the stable phase information and phase difference data with a zero mean, respectively, but none of these are useful for breathing rate estimation. 

To prove the stability of measured CSI phase difference in the 5 GHz band, we write the measured phase of subcarrier $i$, denoted as $\angle{\widehat{\mbox{H}}_i}$, as~\cite{Speth99}
\begin{eqnarray} \label{eq:angle0}
	\angle{\widehat{\mbox{H}}_i}=\angle{\mbox{H}_i} + (\lambda_p+\lambda_s)m_i + \lambda_c + \beta + Z, 
\end{eqnarray}
where $\angle{\mbox{H}_i}$ is the true phase of CSI data, $m_i$ is the subcarrier index of subcarrier $i$, $\beta$ is the initial phase offset at the phase-locked loop (PLL), $Z$ is the measurement environment noise, and $\lambda_p$, $\lambda_s$ and $\lambda_c$ are the phase errors from the packet boundary detection (PBD), the sampling frequency offset (SFO), and central frequency offset (CFO), respectively~\cite{Speth99}, which are given by
\begin{eqnarray} \label{eq:phase errors}
\left\{   \begin{array}{l}
\lambda_p=2 \pi \frac{\Delta t}{N}  \\
\lambda_s=2 \pi (\frac{T'-T}{T}) \frac{T_s}{T_u} n \\
\lambda_c=2 \pi \Delta f T_s n, 
		              \end{array} \right. 
\end{eqnarray}
where $\Delta t$ is the packet boundary detection delay, $N$ is the FFT size, $T'$ and $T$ are the sampling periods from the receiver and the transmitter, respectively, $T_u$ is the length of the data symbol, $T_s$ is the total length of the data symbol and the guard interval, $n$ is the sampling time offset for the current packet, and $\Delta f$ is the center frequency difference between the transmitter and receiver. In fact, the values of $\Delta{t}$, $\frac{T'-T}{T}$, $n$, $\Delta{f}$, and $\beta$ in~(\ref{eq:angle0}) and~(\ref{eq:phase errors}) are unknown, and the values of $\lambda_p$, $\lambda_s$, and $\lambda_c$ can be different for different packets. Thus, we cannot obtain the true phase $\angle{\mbox{H}_i}$ of CSI data from measured phase values. 

However, the measured phase difference on subcarrier $i$ is stable, which can be employed for breathing rate estimation. Since the three antennas (radios) of the Intel 5300 NIC are on the same NIC, they use the same system clock and the same down-converter frequency. The measured CSI phases on subcarrier $i$ from two adjacent antennas have the same $\lambda_p$, $\lambda_s$, $\lambda_c$, and $m_i$. The phase difference can be computed as
\begin{eqnarray} \label{eq:angle1}
  \Delta{\angle{\widehat{\mbox{H}}_i}}=\Delta{\angle{\mbox{H}_i}}+\Delta{\beta}+\Delta{Z}, 
\end{eqnarray}
where $\Delta{\angle{\mbox{H}_i}}$ is the true phase difference of subcarrier $i$, $\Delta{\beta}$ is the unknown difference in phase offsets, which is a constant~\cite{Phaser}, and $\Delta{Z}$ is the noise difference. 
Since in~\eqref{eq:angle1}, the random values $\Delta t$, $\Delta f$, and $n$ are all removed, the phase difference becomes more stable for back-to-back received packets. 
As an example, we plot in Fig.~\ref{D} the phase differences (marked as red dots) and the single antenna phases (marked as gray crosses) read from the $3$rd subcarrier 
for 500 consecutively received packets. It can be seen that the single antenna phase 
is nearly uniformly distributed between 0$^{\circ}$ and 360$^{\circ}$. However, all the phase difference data 
concentrate in a small sector between 330$^{\circ}$ and 340$^{\circ}$, which is significantly more stable than phase data.

\begin{figure} 
\centerline{\includegraphics[width=3.5in]{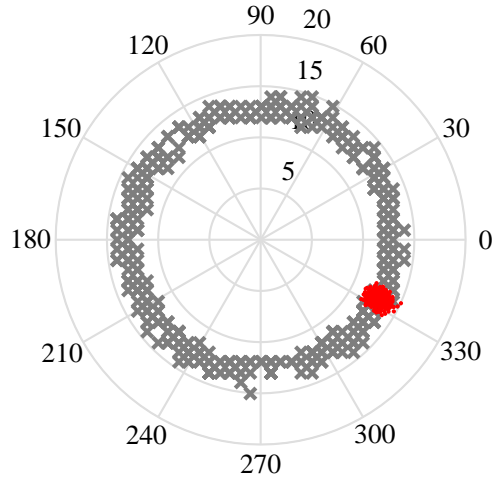}}
\caption{The phase differences (marked by red dots) and the single antenna phases (marked by gray crosses) of the $3$rd subcarrier for 500 back-to-back packets, plotted in the polar coordinate system.}
\label{D}
\end{figure}

Breathing rate estimation for multiple persons is a challenging problem, 
because the reflected components in the received signal are from the chests of multiple persons, each moves slightly due to breathing and the movements are independent. 
Thus, the peak-to-peak detection method cannot be effective for detecting the multiple breathing signals from the received signal.
The aggregated breathing signal from multiple persons is not a clearly periodic signal anymore. Fig.~\ref{motivation1} shows the detected breathing signals for one person (the upper plot) and three persons (the lower plot). 
We can see that for one person, the breathing signal exhibits a noticeable periodicity. So the breathing rate can be estimated by peak detection after removing the noise. However, the aggregated breathing signal of three persons
does not show noticeable periodicity for packet 400 to 600. Traditional FFT based methods can transform the received signal from the time domain to the frequency domain to estimate the breathing frequencies from multiple persons. Fig.~\ref{motivation2} shows the breathing rate estimation for one person (the upper plot) and three persons (the lower plot) with the FFT method. We can see that the estimated frequency for one person is 0.2 Hz, which is almost the same as the true breathing rate. However, for three-person breathing rate estimation, the FFT 
curve only has two peaks, and the estimated breathing rates are much less accurate. In particular, the third peak cannot be estimated. This is because FFT based methods require a larger window size to improve the frequency resolution. 
We show that the proposed tensor decomposition based method is highly effective for multi-person breathing rate estimation in the following section.

\begin{figure} 
\centerline{\includegraphics[width=3.5in]{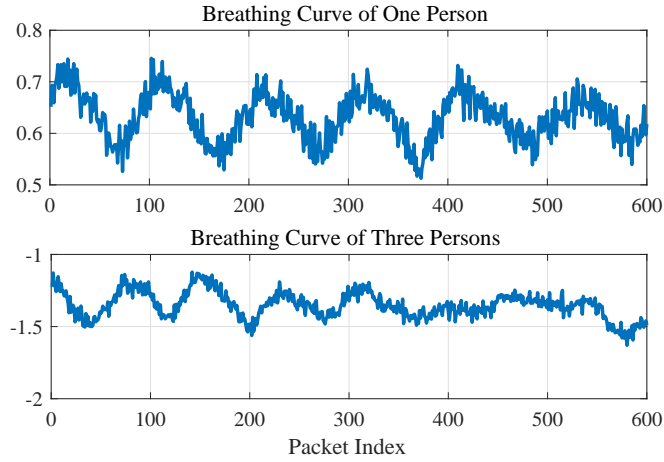}}
\caption{Detected breathing signals for one person (the upper plot) and three persons (the lower plot).}
\label{motivation1}
\end{figure}

\begin{figure} 
\centerline{\includegraphics[width=3.5in]{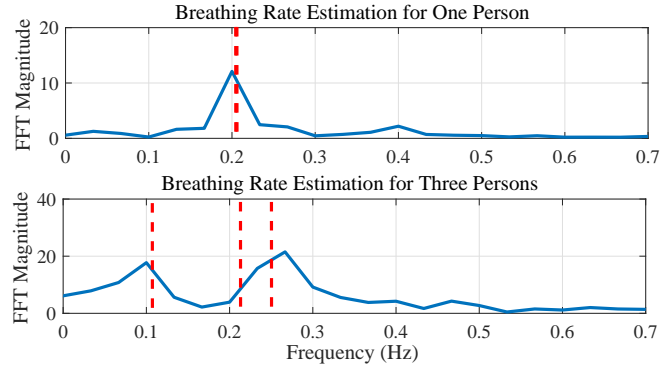}}
\caption{Breathing rate estimation for one person (the upper plot) and three persons (the lower plot) based on FFT.}
\label{motivation2}
\end{figure}

\section{The TensorBeat System}\label{sec:ad}

\subsection{TensorBeat System Architecture}

The main idea of the proposed TensorBeat system is to estimate multi-person breathing rates by employing a tensor decomposition method. To obtain CSI tensor data, we first create a two dimensional Hankel matrix with phase difference data from back-to-back received packets extracted from each subcarrier at each antenna. 
Then, by leveraging the phase differences from the 60 subcarriers, 
i.e., that between antennas 1 and 2, and between antennas 2 and 3, we can construct the third dimension of the CSI tensor data. The TensorBeat system will then leverage the created CSI tensor to estimate multi-person breathing signs. 
Our approach is motivated by two observations. First, for stationary modes of a person, such as standing, sitting, or sleeping, CSI phase difference from consecutively received packets is highly stable. It can thus be useful for extracting the periodic breathing signals. Second, the tensor decomposition method can effectively estimate multi-person breathing beats. We create the CSI tensor data by increasing the dimension of CSI data, from one dimension to three dimensions. The higher dimension CSI data is helpful to effectively separate different breathing signals by forming different clusters. This strategy is similar to the kernel method in traditional machine learning, such as support vector machine (SVM)~\cite{SVM} or multiple hidden layers in deep learning~\cite{learning,DeepFi,DeepFi_J}.

As shown in Fig.~\ref{system}, the TensorBeat system consists of four main modules: Data Extraction, Data Preprocessing, 
CP Decomposition, Signal Matching, and Breathing Rate Estimation. For Data Extraction, TensorBeat obtains 60 CSI phase difference data, 30 between antennas 1 and 2, and 30 between antennas 2 and 3, at the receiver with an off-the-shelf WiFi device. The Data Preprocessing module includes data calibration and Hankelization. Data calibration is implemented to remove the DC component and high frequency noises. Hankelization is to create a two dimensional Hankel matrix with phase difference data from each subcarrier for back-to-back received packets. 
The rank of the constructed Hankel matrix is then analyzed. 
We next apply CP decomposition to estimate multiple persons' breathing signals, and prove the uniqueness of the proposed CSI tensor. For Signal Matching, we first compute the autocorrelation function of the decomposed signals, and incorporate a stable roommate matching algorithm to identify the decomposed signal pairs for each person, where a preference list is computed with the dynamic time warping (DTW) values of the autocorrelation signals. For Breathing Rate Estimation, we combine the decomposed signals in each pair and use the peak detection method to compute the breathing rate for each person. 

\begin{figure} 
\centerline{\includegraphics[width=3.5in]{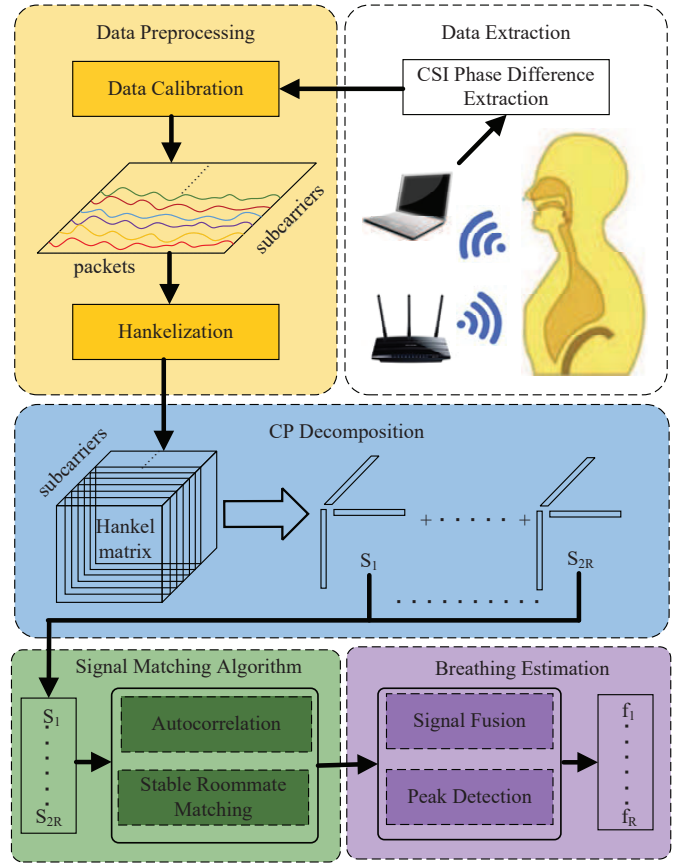}}
\caption{The TensorBeat system architecture.}
\label{system}
\end{figure}

In the remainder of this section, we present the design and analysis of each module of the TensorBeat system in detail. 

\subsection{Data Preprocessing}

\subsubsection{Data Calibration}

We use a 20 Hz sampling rate to obtain 60 CSI phase difference data, 30 between antennas 1 and 2, and 30 between antennas 2 and 3, at the receiver with an off-the-shelf WiFi device at 5 GHz for data extraction. Then, data calibration is applied to remove the DC component and high frequency noises. Because the DC component is also considered as a kind of signal, which may affect CSI tensor decomposition, TensorBeat adopts the Hampel filter to remove the DC component. 
Unlike traditional data calibration approaches that only remove the high frequency noise, we use the Hampel Filter for detrending the original CSI phase difference data to remove DC component. In fact, the Hampel Filter, which is set as a large sliding window with 150 samples wide and a small threshold of 0.001, is firstly used to extract the basic trend of the original data. Then, the detrended data is generated by subtracting the basic trend data from the original data. We also utilize the Hampel Filter to reduce the high frequency noise by using a sliding window of 6 samples wide and a threshold of 0.01. 

Fig.~\ref{calibration} presents an example of data calibration. 
We can see that the original phase differences of all the subcarriers have both a DC component and high frequency noises. With the proposed data calibration approach, it can be seen that the DC components are readily removed and all the subcarriers demonstrate a similar calibrated signal over the 600 packet range with low noise. Such calibrated signal will then be used for estimating the breathing rates of multiple persons.

\begin{figure} [!t]
\centerline{\includegraphics[width=3.5in]{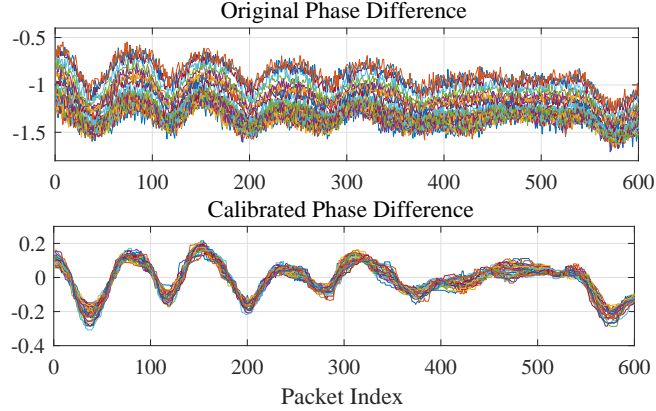}}
\caption{Data calibration: an example.}
\label{calibration}
\end{figure}

\subsubsection{Hankelization}

After data calibration, we obtain the CSI phase difference data matrix with a dimension of (number of packets $\times$ number of subcarriers). 
We then employ a Hankelization method to transform the large CSI matrix into a CSI tensor by expanding the packets into an additional dimension~\cite{tensor3}. Specifically, we rearrange the signals of each subcarrier into a 2-D Hankel matrix, so that the signals from all the 60 subcarriers can be considered as a 3-Dimensional tensor. Define $H_r$ as the constructed Hankel matrix with the size $I \times J$ for subcarrier $r$, which is created by mapping $N$ packets onto the Hankel matrix with $N=I+J-1$. 
We consider the Hankel matrix with size $I=J=\frac{N+1}{2}$. We thus obtain the Hankel matrix $H_r$ for subcarrier $r$, as
\begin{eqnarray}
{\bm H}_r=\begin{bmatrix}
    h_r(0) & h_r(1) & \dots  & h_r(\frac{N+1}{2}-1) \\
    h_r(1) & h_r(2) & \dots  & h_r(\frac{N+1}{2}) \\
    \vdots & \vdots & \vdots & \vdots \\
    h_r(\frac{N+1}{2}-1) & h_r(\frac{N+1}{2})  & \dots  & h_r(N-1)
\end{bmatrix},
\end{eqnarray}
where $h_r(i)$ is the calibrated phase difference data from subcarrier $r$ for packet $i$. In our experiments, we set $N=599$ and $I=J=300$. To determine the number of components needed for CSI tensor decomposition, we provide the following theorem for estimating $R$ breathing signals.

\begin{theorem} \label{th:rank}
If there are $R$ breathing signals in an indoor monitoring environment, the constructed Hankel matrix ${\bm H}_r$ for subcarrier $r$ has a rank of $2R$ when noise is negligible.
\end{theorem}

\begin{proof}
When analyzing signal data structure, we assume the noise is negligible. Moreover, let the $i$th breathing signal be represented as $S_i(t)=A_i\cos(w_i t + \varphi_i)$. 
The observed signal from a subcarrier can be represented by~\cite{tensor2}
\begin{eqnarray} \label{eq:H1}
    Y(t)=\sum_{i=1}^{i=R}K_iS_i(t)=\sum_{i=1}^{i=R}\hat{K_i}\cos(w_i t + \varphi_i),     
\end{eqnarray} 
where $K_i$ is the coefficient for breathing signal $i$ and the new coefficient 
$\hat{K_i}=K_i A_i$. 
The $i$th component of $Y(t)$, $\hat{K_i}\cos(w_i t + \varphi_i)$, can be decomposed using Euler's formula. We have 
\begin{eqnarray} \label{eq:H2}
\hat{K_i}\cos(w_i t + \varphi_i)
     && =\frac{\hat{K_i}}{2}\exp(j(w_i t + \varphi_i))+\frac{\hat{K_i}}{2}\exp(j(-w_i t  -\varphi_i)) \nonumber \\
     && =\frac{\hat{K_i}}{2}\exp(j\varphi_i)\exp(j w_i t)+\frac{\hat{K_i}}{2}\exp(-j\varphi_i)\exp(-j w_i t).
\end{eqnarray} 
Each breathing signal can be separated into two exponential signals with different coefficients. Combining all the $R$ breathing signals, we have 
\begin{eqnarray} \label{eq:H3}
    Y(t)&& = \sum_{i=1}^{R} \left( \frac{\hat{K_i}}{2}\exp(j\varphi_i)\exp(j w_i t)+\frac{\hat{K_i}}{2}\exp(-j\varphi_i)\exp(-j w_i t) \right) \nonumber \\
		    && = \sum_{i=1}^{2R} \tilde{K_i} Z_i^t,
\end{eqnarray} 
where the updated signal $Z_i^t$ is denoted as $Z_i^t=\exp(\pm j w_i t)$, and $\tilde{K_i}=\frac{\hat{K_i}}{2}\exp(\pm j\varphi_i)$ is its coefficient. For packets received at discrete times, we represent the received signal as $Y(n)=\sum_{i=1}^{2R} \tilde{K_i} Z_i^n$. Note that the combined signal can be considered as an exponential polynomial with $2R$ different exponential terms. Map signal $Y(n)$ for $n=1, 2, \cdots, N$ into a Hankel matrix with size $I=J=\frac{N+1}{2}$, we have  
\begin{eqnarray}
{\bm H}_r=\begin{bmatrix}
    \sum_{i=1}^{2R}\tilde{K_i}Z_i^0 & \sum_{i=1}^{2R}\tilde{K_i}Z_i^1 & \cdots  & \sum_{i=1}^{2R}\tilde{K_i}Z_i^{\frac{N+1}{2}-1} \\
    \sum_{i=1}^{2R}\tilde{K_i}Z_i^1 & \sum_{i=1}^{2R}\tilde{K_i}Z_i^2 & \cdots  & \sum_{i=1}^{2R}\tilde{K_i}Z_i^{\frac{N+1}{2}} \\
    \vdots & \vdots & \cdots & \vdots \\
    \sum_{i=1}^{2R}\tilde{K_i}Z_i^{\frac{N+1}{2}-1} & \sum_{i=1}^{2R}\tilde{K_i}Z_i^{\frac{N+1}{2}} & \cdots  & \sum_{i=1}^{2R}\tilde{K_i}Z_i^{N-1}
\end{bmatrix}. 
\end{eqnarray}
We can see that the Hankel matrix can be decomposed with  
Vandermonde decomposition~\cite{tensor3}, as
\begin{eqnarray} \label{eq:H4}
    {\bm H}_r= {\bm V}_r \cdot \mbox{diag} (\tilde{K_1},\tilde{K_1}, \cdots, \tilde{K}_{2R}) \cdot \tilde{\bm V}_r^T,    
\end{eqnarray} 
where the Vandermode matrices ${\bm V}_r \in \mathbb{K}^{\frac{N+1}{2} \times 2R}$ and $\tilde{\bm V}_r \in \mathbb{K}^{\frac{N+1}{2} \times 2R}$ are given by
\begin{eqnarray}
{\bm V}_r=\tilde{\bm V}_r=\begin{bmatrix}
    1 & 1 & \cdots & 1 \\
    Z_1 & Z_2 & \cdots  & Z_{2R} \\
    \vdots & \vdots & \cdots & \vdots \\
    Z_1^{\frac{N+1}{2}-1} & Z_2^{\frac{N+1}{2}-1} & \cdots  & Z_{2R}^{\frac{N+1}{2}-1}
\end{bmatrix}.
\end{eqnarray}
Because a Vandermode matrix is full rank, which is obtained by different poles, the rank of the Hankel matrix generated by $R$ breathing signals is $2R$. 
\end{proof}

According to Theorem~\ref{th:rank}, $2R$ signal components is required to separate the $R$ breathing signals. 

Next we consider the influence of measurement noise on the Hankel matrix ${\bm H}_r$. Because of 
noise, the Hankel matrix $H_r$ is actually a full-rank matrix. However, Theorem~\ref{th:rank} shows that the rank of the combined breathing signal is $2R$, meaning that the first $2R$ weighted decomposed components are much stronger than the remaining ones as long as the signal to noise ratio (SNR) is not very low. This shows that the Hankel matrix structure can be used to effectively separate breathing signals from white noise. Actually, the different signals will be well denoised and separated by using tensor decomposition, as to be discussed in Section~\ref{subsec:cpd}.

\subsection{Canonical Polyadic Decomposition \label{subsec:cpd}}

Once the CSI tensor is ready, we apply CP decomposition to estimate multiple persons' breathing signals. With CP decomposition, the CSI tensor data can be approximated as the sum of $2R$ rank-one tensors according to Theorem~\ref{th:rank}. Denote $\chi \in \mathbb{K}^{I \times J \times K} $ as a third-order CSI tensor, which can be obtained by the sum of three-way outer products as~\cite{tensor1,tensor4}
\begin{eqnarray} \label{eq:CP1}
    \chi \approx \sum_{r=1}^{2R}a_r \circ b_r \circ c_r, 
\end{eqnarray}
where $a_r$, $b_r$, $c_r$ are the vectors at the $r$th position for the first, second, and third dimension, respectively, and $2R$ is the number of decomposition components, which is the approximation rank of the tensor based on CP decomposition~\cite{rank1,rank2}. Their outer product is defined by
\begin{eqnarray} \label{eq:CP2}
    (a_r \circ b_r \circ c_r) (i,j,k)=a_r(i)b_r(j)c_r(k), \;\;\; \mbox{for all} & i,j,k.
\end{eqnarray}
We consider factor matrices ${\bm A}=[a_1,a_2, \cdots, a_{2R}] \in \mathbb{K}^{I \times 2R}$, ${\bm B}=[b_1,b_2, \cdots, b_{2R}] \in \mathbb{K}^{J \times 2R}$, and ${\bm C}=[c_1,c_2, \cdots, c_{2R}] \in \mathbb{K}^{K \times 2R}$ as the combination of vectors from rank-one components. Moreover, define ${\mathcal X}_{(1)} \in \mathbb{K}^{I \times JK}$, ${\mathcal X}_{(2)} \in \mathbb{K}^{J \times IK}$, and ${\mathcal X}_{(3)} \in \mathbb{K}^{K \times IJ}$ as 1-mode, 2-mode, and 3-mode matricization of CSI tensor $\chi \in \mathbb{K}^{I \times J \times K}$, respectively, which are obtained by fixing one mode and arranging the slices of the rest of the modes into a long matrix~\cite{tensor1}. Then, we can write the three matricized forms as
\begin{eqnarray} \label{eq:CP3}
    {\mathcal X}_{(1)} \approx {\bm A} ( {\bm C} \odot {\bm B})^T,\\
		{\mathcal X}_{(2)} \approx {\bm B} ( {\bm C} \odot {\bm A})^T,\\
		{\mathcal X}_{(3)} \approx {\bm C} ( {\bm B} \odot {\bm A})^T,
\end{eqnarray}
where $\odot$ denotes the Khatri-Rao product.

When the number of components $2R$ is given, we apply the Alternating Least Squares (ALS) algorithm, the most widely used algorithm for CP decomposition~\cite{tensor1}. To decompose the CSI tensor, we minimize the square sum of the differences between the CSI tensor $\chi$ and the estimated tensor. 
\begin{eqnarray} \label{eq:CP4}
    \min_{{\bm A}, {\bm B}, {\bm C}}\left\| \chi- \sum_{r=1}^{2R}a_r \circ b_r \circ c_r \right\|_F^2 .
\end{eqnarray} 
Note that~\eqref{eq:CP4} is not convex. However, the ALS algorithm can effectively solve the problem by fixing two of the factor matrices, to reduce the problem to a linear least squares problem with the third factor matrix as variable. If we fix $\bm B$ and $\bm C$, we can rewrite problem~(\ref{eq:CP4}) as 
\begin{eqnarray} \label{eq:CP5}
    \min_{{\bm A}}\left\| {\mathcal X}_{(1)}- {\bm A} ( {\bm C} \odot {\bm B})^T\right\|_F^2 .
\end{eqnarray} 
We can derive the optimal solution to problem~\eqref{eq:CP5} as ${\bm A}={\mathcal X}_{(1)}[( {\bm C} \odot {\bm B})^T]^{\dagger}$. Applying the property of pseudoinverse of the Khatri-Rao product, it follows that 
\begin{eqnarray} \label{eq:CP6}
    {\bm A}={\mathcal X}_{(1)}( {\bm C} \odot {\bm B})( {\bm C}^{T} {\bm C} \ast {\bm B}^{T} {\bm B})^{\dagger}, 
\end{eqnarray} 
where $\ast$ denotes the Hadamard product. This equation only requires computing the pseudoinverse of a $2R \times 2R$ matrix rather than a $JK \times 2R$ matrix. Note that $R$ is much smaller than $J$ and $K$, thus the computing complexity can be greatly reduced. Similarly, we can obtain the optimal solutions for ${\bm B}$ and ${\bm C}$ as
\begin{eqnarray} \label{eq:CP7}
    {\bm B}={\mathcal X}_{(2)}({\bm C} \odot {\bm A})({\bm C}^{T} {\bm C} \ast {\bm A}^{T} {\bm A})^{\dagger} \\
		{\bm C}={\mathcal X}_{(3)}({\bm B} \odot {\bm A})({\bm B}^{T} {\bm B} \ast {\bm A}^{T} {\bm A})^{\dagger}.
\end{eqnarray} 
Applying ALS to CP decomposition, we obtain matrices $\bm A$, $\bm B$, and $\bm C$. To guarantee the effectiveness of the decomposed components, 
we next examine 
the uniqueness of CP decomposition. The basic theorem on the uniqueness of CP decomposition is 
given in~\cite{tensor1}, which is provided in the following.

\begin{fact}\label{unique1}
For tensor $\chi$ with rank $L$, if $k_A+k_B+k_C \geq 2 L + 2$, then the CP decomposition of $\chi$ is unique,
where $k_{\bm A}$, $k_{\bm B}$, and $k_{\bm C}$ denote the $k$-rank of matrix ${\bm A}$, ${\bm B}$, ${\bm C}$, respectively. Here $k$-rank means the maximum value $k$ such that any $k$ columns are linearly independent~\cite{tensor1}.
\end{fact}

Based on Fact~\ref{unique1}, we have the following theorem for the CSI tensor. 

\begin{theorem} \label{unique2}
For the proposed CSI tensor $\chi$ with rank $2R$, the CP decomposition of $\chi$ is unique.
\end{theorem}

\begin{proof}
The proposed CSI tensor $\chi$ is created by $K$ Hankel matrix, where the $r$th Hankel matrix ${\bm H}_r$ is rank-2R according to Theorem~\ref{th:rank}. Thus, for the $k$-rank of the matrices ${\bm A}$ and ${\bm B}$, we have $k_{\bm A}=2R$ and $k_{\bm B}=2R$. On the other hand, because phase differences of subcarriers between antennas 1 and 2, and antennas 2 and 3 are independent, the $k$-rank of matrix $\bm C$ has $k_{\bm C} \geq 2$. Thus, the expression is $k_{\bm A} + k_{\bm B} + k_{\bm C} \geq 2R+2R+2=2 (2R)+2$, which satisfies the conditions in Theorem~\ref{unique1}. This proofs the theorem.
\end{proof}

Theorem~\ref{unique2} indicates that the CP decomposition of the created CSI tensor is unique, which can be used to effectively estimate multiple breathing rates. In the proposed TensorBeat system, we leverages the matrix ${\bm A}=[a_1, a_2, \cdots, a_{2R}]$ as decomposed signals $S_1, S_2, \cdots, S_{2R}$. For example, Fig.~\ref{CP} shows the results of CP decomposition for CSI tensor data from three persons ($R=3$). We can see that there are six signals. Moreover, signals 1 and 2 are similar, signals 3 and 5 are similar, and signals 4 and 6 are similar. This is because CP decomposition cannot guarantee that similar signals are located in adjacent locations (i.e., the output signals are randomly indexed). Thus, we need to identify the signal pairs among the decomposed signals for each person, which will be addressed in Section~\ref{subsec:sma}. 

\begin{figure} [!t]
\centerline{\includegraphics[width=3.5in]{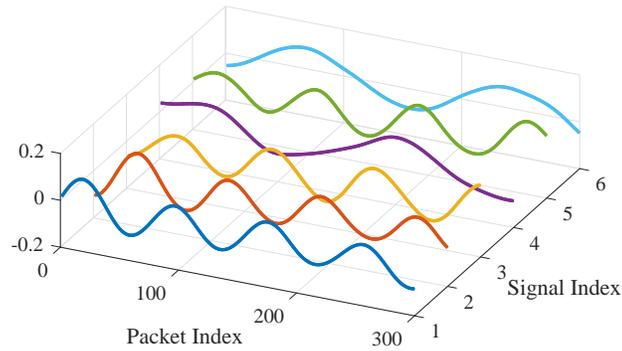}}
\caption{CP decomposition results for a CSI tensor of three persons.}
\label{CP}
\end{figure}

\begin{figure} [!t]
\centerline{\includegraphics[width=3.5in]{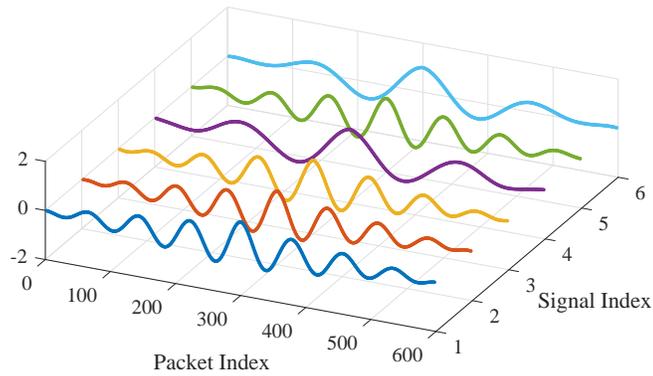}}
\caption{Autocorrelation of the decomposed breathing signals.}
\label{auto}
\end{figure}

\subsection{Signal Matching Algorithm \label{subsec:sma}}

The CP decomposition of CSI tensor data yields $2R$ decomposed signals, i.e., $S_1, S_2, \cdots, S_{2R}$, which, however, are randomly indexed. In this section, we propose a signal matching algorithm to pair the two similar decomposed signals that belong to the same person. The main idea is to leverage the autocorrelation to strengthen the periodicity of decomposed signals and use the Dynamic Time Warping (DTW) method to compute the similarity value for any pair of signals. Finally, we apply the stable-roommate matching algorithm to pair the decomposed signals for each person, using the DTW values as the closeness metric. 
We introduce the proposed signal matching algorithm in the following.

\subsubsection{Autocorrelation and Dynamic Time Warping}

After CP decomposition of CSI tensor data, we first compute the autocorrelation function of the $2R$ decomposed signals to strengthen their periodicity. 
We evaluate the autocorrelation function of the decomposed signals for two reasons. The first is that the autocorrelation of a decomposed signal can increase the data length, 
which 
helps to improve the accuracy of the peak detection. Second, because the decomposed signals have phase shift and nonalignment, using the autocorrelation of decomposed signals can reduce such shifts and strengthen the periodicity of the decomposed signals. Fig.~\ref{auto} shows the autocorrelation of the decomposed breathing signals produced by CP decomposition. We can see that each 
autocorrelation signal exhibits a more obvious periodicity than that of the original decomposition signals. Moreover, the data length 
is increased from 300 to 600. 

Furthermore, we employ the DTW approach to measure the distance between any pair of autocorrelation signals, which is different from the Euclidean distance method that computes the sum of distances from each value on one curve to the corresponding value on the other curve. 
Moreover, the Euclidean distance method believes that two autocorrelation signals with the same length are different as long as one of them has a small shift. However, DTW can automatically identify these shifts and provide the similar distance measurement between two autocorrelation signals by aligning the corresponding time series, thus overcoming the limitation of the Euclidean distance method. 

With the autocorrelation signals, we design the DTW method for measuring their pairwise distance. Given two autocorrelation signals and a cost function, the DTW method seeks an alignment by matching each point in the first autocorrelation signal to one or more points in the second signal, thus minimizing the cost function for all points~\cite{Dude,DTW1,DTW2}. To reduce the computational complexity of DTW, we apply downsampling to the two autocorrelation signals, which leads to a reduced number of packets $N'$. Then, consider two downsampled autocorrelation signals $P_i=[P_i(0),P_i(1),\cdots,P_i(N'-1)]$ and $P_j=[P_j(0),P_j(1),\cdots,P_j(N'-1)]$, we need to find a warping path $W=[w_1, w_2, \cdots, w_L]$, where $L$ is the length of the path, and the $l$th element of the warping path is $w_l=(m_l,n_l)$, where $m$ and $n$ are the packet index for the two downsampled autocorrelation signals. The objective is to minimize the total cost function by implementing the non-linear mapping between two downsampled autocorrelation signals $P_i$ and $P_j$. The formulated problem is given by  
\begin{eqnarray} \label{DTW1}
    \min  && \sum_{l=1}^{L}\left\| P_i(m_l)-P_j(n_l) \right\| \\
		\mbox{s.t.} && (m_1,n_1)=(0,0)\\
		            && (m_L,n_L)=(N'-1,N'-1)\\
				        && m_l \leq m_{l+1} \leq m_l+1\\
								&& n_l \leq n_{l+1} \leq n_l+1. 
\end{eqnarray} 

The objection function is to minimize the distance between two downsampled autocorrelation signals. 
The first and second constraints are boundary constraints, which require that the warping path starts at $P_i(0)$ and $P_j(0)$ and ends at $P_i(N'-1)$ and $P_j(N'-1)$. This can guarantee all points of the two downsampled autocorrelation signals are used for measuring their distance, thus avoiding to use only local data. Furthermore, the third and fourth constraints are monotonic and marching constraints, which require that there be no cycles for $w_i$ and $w_j$ in the warping path and the path is increased with the maximum 1 at each step. 

We apply dynamic programming to solve problem~\eqref{DTW1}, to obtain the minimum distance warping path between two  downsampled autocorrelation signals. We consider a two-dimensional cost matrix $\mathcal{C}$ with size $N' \times N'$, whose element $\mathcal{C}(m_l,n_l)$ is the minimum distance warping path for two downsampled autocorrelation signals $P_i=[P_i(0),P_i(1), \cdots, P_i(m_l)]$ and $P_j=[P_j(0),P_j(1),\cdots,P_j(n_l)]$. We design the recurrence equation in dynamic programming as follows.
\begin{eqnarray} \label{DTW2}
    \mathcal{C}(m_l,n_l) = 
		            \left\| P_i(m_l)-P_j(n_l) \right\| + 
								\min{[\mathcal{C}(m_l-1,n_l), \mathcal{C}(m_l,n_l-1), \mathcal{C}(m_l-1,n_l-1)]}. 
\end{eqnarray}
By filling all elements of the cost matrix $\mathcal{C}$, the value $\mathcal{C}(N'-1,N'-1)$ can be computed as the DTW value between the two downsampled autocorrelation signals. The time complexity is $O(N'^2)$. Fig.~\ref{DTW} shows the DTW results for downsampled autocorrelation signals 4 and 6 (the upper plot), and downsampled autocorrelation signals 4 and 3 (the lower plot), where we set the downsampling number of packets as $N'=\frac{N}{10}=60$. It can be seen that downsampled autocorrelation signals 4 and 6 have a smaller DTW value (i.e., 3.65) than downsampled autocorrelation signals 4 and 3 (i.e., 13.7). That is, signals 4 and 6 are more similar, and more likely to belong to the same person. We also find that the downsampled autocorrelation signals have a high similarity in the center than that on the boundary, and it can reduce the phase shift values. Thus, the DTW value 
is a good measure of the distance between two downsampled autocorrelation signals. We need to compute the DTW values for all the downsampled autocorrelation signal pairs, which are then used in stable roommate matching. 

\begin{figure} [!t]
\centerline{\includegraphics[width=3.5in]{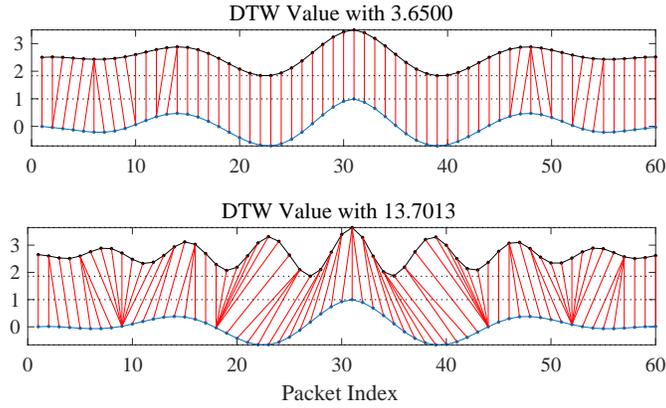}}
\caption{DTW results for downsampled autocorrelation signals 4 and 6 (the upper plot), and downsampled autocorrelation signal 4 and 3 (the lower plot), respectively.}
\label{DTW}
\end{figure}

\subsubsection{Stable Roommate Matching}

Since the CP decomposed signals are randomly indexed (see Fig.~\ref{auto}), we need to identify the pair for each person. With the DTW values for all downsampled autocorrelation signal pairs, we can model this problem as a stable roommate matching problem~\cite{matching1,matching2,matching3}. There are a group of $2R$ signals, and each signal maintains a preference list of all other signals in the group, where the preference value for another signal is the inverse of the corresponding DTW value (i.e., distance). The problem is to pair the signals, such that there is no such a pair of signals that both of them have a more desired selection than their current selection, i.e., to find a stable matching~\cite{matching1,matching2,matching3}. The proposed signal matching algorithm is presented in Algorithm~\ref{alg:one}. 


\begin{algorithm}[t]
\SetAlgoLined
\KwIn{Decomposed signals: $S_1, S_2, \cdots, S_{2R}$.}
\KwOut{Matched signal pairs.}
Compute autocorrelation of all decomposed signals\;
Compute the DTW values for every pair of autocorrelation signals\;
Each autocorrelation signal sets its preference list using the DTW values\;
//step1\;
\For{$signal\_num=1:2R$}{
Set $finish\_flag=0$\; 
Set $scan\_num=signal\_num$\;
    \While{$finish\_flag=0$}{
		    \eIf{the proposal is the first one}{
           Proposing signal's $propose\_num$=the current choice\;
					 Set $finish\_flag=1$\;
					 Proposed signal's $accept\_num=scan\_num$\;
         }{
           \eIf{the signal prefers the former proposal}{
           Reject the current proposal symmetrically\;
					 Propose to the next choice\;
           }{
           Accept the current proposal\;
					 Reject the former proposal symmetrically\;
					 $scan\_num=$ proposed signal's $accept\_num$\;
					 Propose to the next choice\;
           }
         }
		}		
}

\For{$ signal\_num=1:2R $}{
	 Reject signals that have less than $accept\_num$ in every preference list symmetrically\; 
}

//step2\;
$signal\_num=1$\;
\While{$signal\_num < 2R+1$}{
   \eIf{$propose\_num=accept\_num$}{
	     $signal\_num$=$signal\_num$+1;   
	     }{
			Let $p_1$ be a signal whose preference list contains more than one element\;
			\While{$p$ sequence is not cyclic}{
			   $q_i$ = the second preference of $p_i$'s current list\;
			   $p_{i+1}$ = the last preference of $q_i$'s current list\;
			}
			Denote $p_s$ as the first element in the $p$ sequence to be repeated and $r$ as the length of the circle\;
			\For{$i=1:r$}{
			      Reject matching $\left( q_{s+i-2}, p_{s+i-1} \right)$ symmetrically\;  
			    }
			$signal\_num=1$\;
			}
		}
	 Obtain signal matching pairs based on all processed preference lists\;
   \caption{Signal Matching Algorithm}
\label{alg:one}
\end{algorithm}

We first compute the autocorrelation of all decomposed signals. Then each autocorrelation signal populates its preference list with other autocorrelation signals according to the DTW values. 
The stable roommate matching algorithm is executed in two steps. In step 1, each signal proposes to other signals according to its preference list. If a signal $m$ receives a proposal from another signal $n$, we implement the following strategy: (i) signal $m$ rejects signal $n$ if it has a better proposal from another signal; (ii) signal $m$ accepts signal $n$'s proposal if it is better than all other proposals that signal $m$ currently holds. Moreover, signal $n$ stops to propose when its proposal is accepted, while it needs to continue to propose to other signals if being rejected. This strategy is implemented in step 1 of the signal matching algorithm, where we use $finish\_flag$ to mark whether the current $signal\_num$ is accepted or not. Moreover, variables $accept\_num$ and $propose\_num$ are used to record the current signal's proposed number and proposing number, respectively. Also, variable $scan\_num$ is used to record the current scanning signal number. After completing step 1, every signal holds a proposal or one signal has been rejected by other signals (this case hardly happens in TensorBeat, because the CP decomposition produces two very similar signals with high probability for each person). Then, we need to delete some elements in all the preference lists based on the following method, which is that if signal $m$ is the first on signal $n$'s list , then signal $n$ is the last on signal $m$'s list. For the proposed algorithm, every signal can reject signals that have less than $accept\_num$ in its preference list symmetrically (reject each other).

An example is shown in Fig.~\ref{CP}. According to the DTW values, signals 1, 2, 3, 4, 5, and 6 have their preference lists as (2, 3, 5, 6, 4), (1, 3, 5, 6, 4), (5, 1, 2, 6, 4), (6, 5, 3, 2, 1), (3, 2, 1, 4, 6), and (4, 5, 3, 2, 1), respectively. When step 1 is executed, we have: Signal 1 proposes to 2, and signal 2 holds 1;  Signal 2 proposes 1, and signal 1 holds; Signal 3 proposes to signal 5, and signal 5 holds; Signal 4 proposes to signal 6, and signal 6 holds; Signal 5 proposes to signal 3, and signal 3 holds; Signal 6 proposes to signal 4, and signal 4 holds. It is easy to find three pairs (1,2), (3,5), and (4,6). 

Although most of decomposed signals are paired in step 1, step 2 will still be necessary for the more challenging cases of  
much more breathing signals and NLOS environments. 
In step 2, we consider the reduced preference lists, where some of the lists have more than one signals. By implementing step 2, we can reduce the preference lists such that each signal only holds one proposal. The main idea is that we need to find some all-or-nothing cycles and symmetrically delete signals in the cycle sequence by rejecting the first and last choice pairs. The signal in the cycle accepts the secondary choice, thus obtaining a stable roommate matching. To find all-or-nothing cycles, let $p_1$ be a signal with a preference list that contains more than one element, and generate the sequences such that $q_i$ = the second preference of $p_i$'s current list, and $p_{i+1}$ = the last preference of $q_{i}$'s current list. After the cycle sequence generation, denote $p_s$ as the first element in the $p$ sequence to be repeated. Then, we reject matching $(q_s+i-2, p_s+i-1)$ for $i=1$ to $r$ symmetrically, where $r$ is the length of the cycle. Finally, we can obtain signal matching pairs based on all processed preference lists. 
The computational complexity of Algorithm~\ref{alg:one} is $\textit{O}(R^2)$, because steps 1 and 2 each has a complexity of $\textit{O}(R^2)$, respectively.

\subsection{Breathing Rate Estimation}

\subsubsection{Signal Fusion and Autocorrelation}

After obtaining the outcomes of the signal matching algorithm, TensorBeat next applies peak detection to estimate the breathing rates for multiple persons. Comparing to the FFT method, a higher resolution in the time domain can be achieved. 
To implement peak detection, we first need to combine the decomposed signal pairs for each person into a single signal, by taking the average of the signal pairs. 
Averaging can decrease the variance of the decomposed signals while preserving the same period. 
For example, Fig.~\ref{fus} shows the fusion results based on the outcome of the signal matching algorithm, where three smoothly decomposed signals with different periods are obtained. To strengthen the accuracy of peak detection, we compute the autocorrelation function again for every fused signal. 
Fig.~\ref{auto2} shows the autocorrelation of the three fused signals. It can be seen that the length of data is increased from 300 to 600 and the number of peaks of every signal are also increased, which help to improve the estimation accuracy. 

\begin{figure} [!t]
\centerline{\includegraphics[width=3.5in]{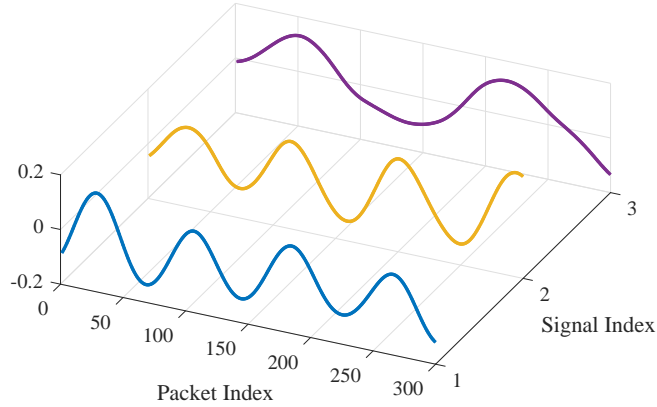}}
\caption{Fusion results based on the outcomes of the signal matching algorithm.}
\label{fus}
\end{figure}

\begin{figure} [!t]
\centerline{\includegraphics[width=3.5in]{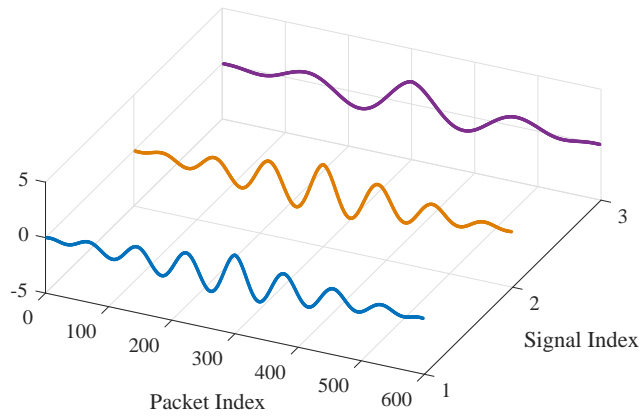}}
\caption{Autocorrelation of fused signals.}
\label{auto2}
\end{figure}

\subsubsection{Peak Detection}

Although breathing signal is generated by the small periodic chest movement of inhaling and exhaling, the phase difference data can effectively capture the breathing rate. Traditionally, estimation of breathing rates is achieved with FFT based methods. However, 
the FFT approach may have limited accuracy, because the frequency resolution of breathing signals is based on the window size of FFT. 
When the window size becomes larger, the accuracy will be higher, but the time domain resolution will be reduced. 
Also see Figs.~\ref{motivation1} and~\ref{motivation2} for the limitation of the FFT based approach for the multi-person scenario. 
Therefore, we leverage peak detection instead in TensorBeat system to achieve accurate breathing rate estimation for each of autocorrelations of fused signals.

For peak detection, the traditional method based on amplitude needs to detect the fake peak, which is not a real peak but has larger values than its two immediate neighboring points. To avoid the fake peak, a large moving window can be used to identify the real peak based on the maximum breathing periodicity. This method is not robust, which requires adjusting the window size. In TensorBeat, we only consider a smaller moving window of 7 samples wide. This is because we leverage the Hankel matrix and CP decomposition to smooth out the breathing curves, which hardly contains any fake peaks. Then, for the $i$th autocorrelation curve of fused signal, we seek all the peaks by determining whether or not the medium of the 7 samples in the moving window is the maximum value. Finally, we consider the median of all peak-to-peak intervals as the final period of the $i$th breathing signal, which is denoted as $T_i$. Finally, the estimated breathing rates can be computed as $f_i=60/T_i$, for $i=1$ to $R$.
 
\section{Experimental Study}\label{sec:sml}

\subsection{Experiment Configuration}

In this section, we validate the TensorBeat performance with an implementation with 5 GHz Wi-Fi devices. 
To obtain 5 GHz CSI data, we use a desktop computer and a Dell laptop as access point and mobile device, respectively, both of which are equipped with an Intel 5300 NIC. We use the desktop computer instead of the commodity routers, because none is equipped with the Intel 5300 NIC. The operating system is Ubuntu destop 14.04 LTS OS for both the access point and the mobile device. The PHY is the IEEE 802.11n OFDM system with QPSK modulation and 1/2 coding rate. Moreover, the access point is set in the monitor model and the distance between its two adjacent antennas is approximately 2.68 cm, which is half of the wavelength of 5GHz WiFi. Also, the mobile device is set in the injection model and uses one antenna to transmit data. Moreover, we use omnidirectional antennas for both the receiver and transmitter to estimate breathing signs beats. 
With the packet injection technique with LORCON version 1, we can obtain 5 GHz CSI data from the three antennas of the receiver.

Our experimental study is with up to five persons over a period of six months. The experimental scenarios include a computer laboratory, a through-wall scenario, and a corridor, as shown in Fig.~\ref{deployment}. The first scenario is within a 4.5 $\times$ 8.8 $m^2$ laboratory, where both single person and multi-person breathing rate estimation experiments are conducted. There are lots of tables and desktop computers crowded in the laboratory, which block parts of the LOS paths and form a complexity radio propagation environment. The second setup is a through-wall environment, where single person breathing rate estimation is tested due to the relatively weaker signal reception. The person is on the transmitter side, and the receiver is behind a wall in this experiment. The third scenario is a long corridor of 20 m, where the maximum distance between the receiver and transmitter is 11 m in the experiment. This scenario is still considered for single person breathing rate monitoring. We use a NEULOG Respiration to record the ground truths for single person breathing rates. The single person breathing rates estimation can be easily implemented by removing the signal matching algorithm, because there are only two decomposed signals after CP decomposition in this case. For muti-person breathing rate estimation in the first scenario, all persons participating in the experiment record their breathing rates by using a metronome smartphone application with 1 bpm accuracy at the same time. We consider five persons are stationary for LOS and NLOS environments for breathing monitoring. Moreover, there are no other persons in the breathing measurement area. 

\begin{figure} [!t]
\centerline{\includegraphics[width=3.5in]{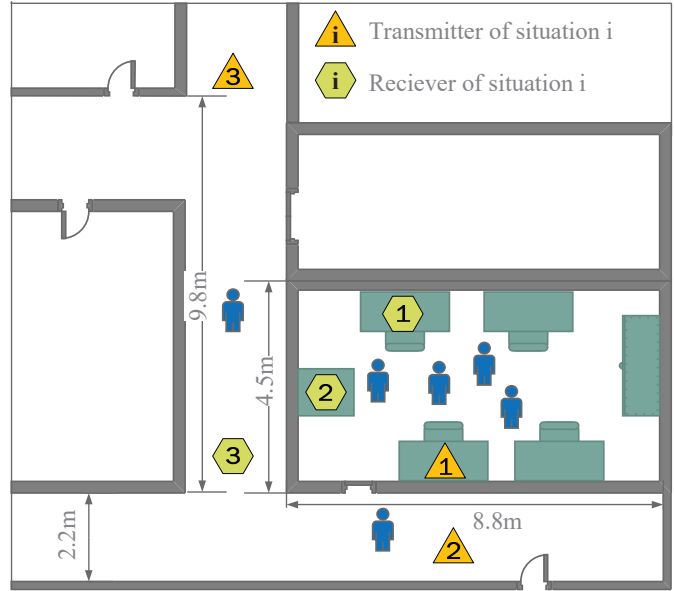}}
\caption{Experimental setup: computer laboratory, through-wall, and long corridor scenarios.}
\label{deployment}
\end{figure}

For multi-person breathing rate estimation, we need to define a proper metric for evaluating TensorBeat's performance. For $R$ estimated breathing rates [$f_1, f_2,...f_R$], the $i$th breathing rate estimation error, $E_i$, is defined as 
\begin{eqnarray} \label{eq:m1}
    E_i=\left|f_i-\hat{f}_i\right|, & \mbox{for } i=1,2,\cdots,R, 
\end{eqnarray}
where $\hat{f}_i$ is the ground truth of the $i$th breathing rate. 
We also define a new metric termed success rate, denoted as $SR$, which is defined as 
\begin{eqnarray} \label{eq:m2}
SR=\frac{N\{\max_i{\{E_i}\} < 2 \mbox{bpm}\}}{N\{E\}} \times 100\%,
\end{eqnarray}
where $N\{\max_i{\{E_i\}} < 2 \mbox{bpm}\}$ means the number of repeated experiments of the maximum breathing rate error less than 2 bpm, and $N\{E\}$ is the number of repeated experiments. We adopt the success rate metric because there are weak signals for multi-person experiments in indoor experiments at different locations, and 
sometimes a breathing signal may not be successfully detected~\cite{when}.

\subsection{Performance of Breathing Estimation}

In Fig.~\ref{breathing}, we present the cumulative distribution functions (CDF) of the estimation errors for single person breathing rate detection for three different experiment scenarios. We can see that for TensorBeat, high estimation accuracy of breathing rates can be achieved in all the three scenarios. The maximum estimation error is less than 0.9 bpm. Moreover, it is noticed that 50\% of the tests for the computer laboratory experiment have errors less than aboout 0.19 bpm, while the tests for the corridor and through-wall scenarios have errors less than approximately 0.25 bpm and 0.35 bpm, respectively. Thus, the performances in the laboratory setting is better than that in the corridor and through-wall scenarios. This is because the laboratory has a smaller space and the breathing signal is stronger than that of other two cases with larger attenuation due to the long distance and the wall. 

\begin{figure} [!t]
\centerline{\includegraphics[width=3.5in]{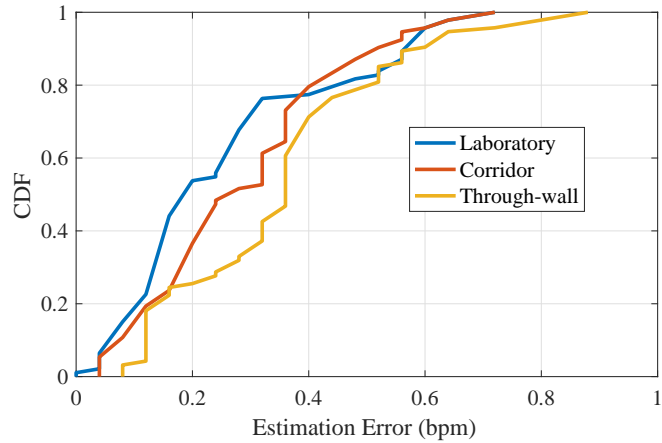}}
\caption{Performance of single person breathing rate estimation in the computer laboratory, through-wall, and long corridor scenarios.}
\label{breathing}
\end{figure}

Fig.~\ref{12345cdf} presents the performance of breathing rate estimation for different number of persons. It is noticed that higher accuracy is achieved for the single person test, where approximately 96\% of the test data have an estimation error less than 0.5 bpm. The five-person test has the worse performance, where approximately 62\% of the test data have an estimation error less than 0.5 bpm. Moreover, we fine that the performances of the two-person and three-person tests are similar, both of which can have an error smaller than 0.5 bpm for 93\% of the test data. Generally, when the number of persons is increased, the performance of breathing rate estimation gets worse. In fact, when the number of breathing signals is increased, the distortion of the mixed received signal will become larger, thus leading to high estimation errors. 

\begin{figure} [!t]
\centerline{\includegraphics[width=3.5in]{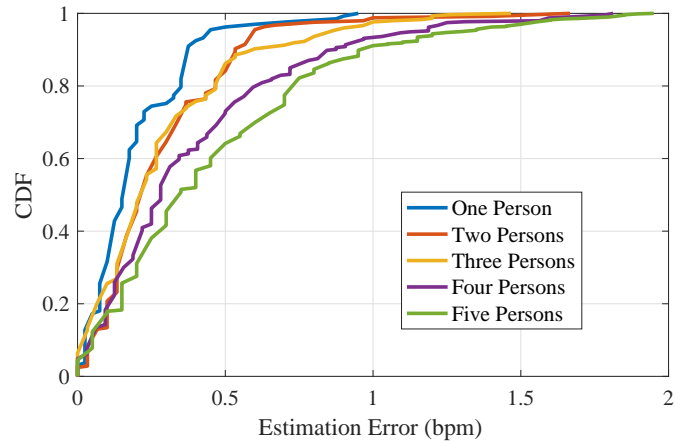}}
\caption{Performance of breathing rate estimation for different number of persons (computer laboratory).}
\label{12345cdf}
\end{figure}

Fig.~\ref{SR_12345} plots the success rates for different number of persons. We find that although the success rate for one person is the highest, there are still few of test data that cannot obtain high accuracy breathing rates estimation. These test data should come from different locations in the indoor environments, where parts of the received signals are severely distorted. In fact, we find that low phase difference usually occurs when the SNR is low. 
On the other hand, we can see that breathing rate estimation for two persons also has a high success rate, because the probability for two persons to have exactly the same breathing rate is very low. 
When the number of persons is increased, the chance of getting two close breathing rates becomes higher. Even in this case, 
TensorBeat can still effectively separate them with a high success probability. With the increase of the number of persons, the success rate for TensorBeat system decreases. The reason is that each breathing rate is more likely to cover each other and the strength of the received signal becomes lower. From Fig.~\ref{SR_12345}, we can see that the success rate is about 82.4\% when the number of persons is five. 

\begin{figure} 
\centerline{\includegraphics[width=3.5in]{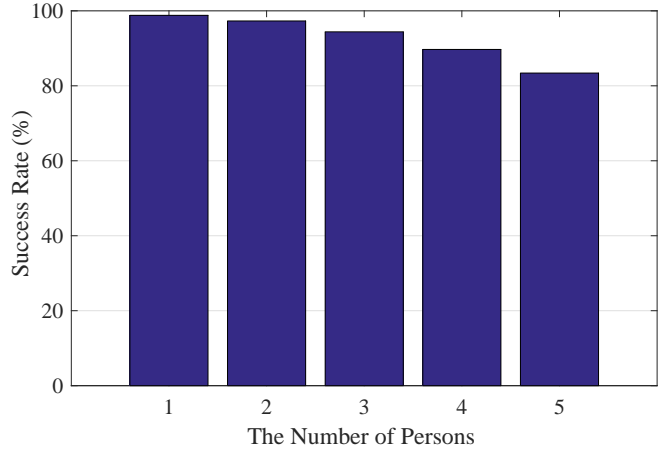}}
\caption{Success rates for different number of persons (computer laboratory).}
\label{SR_12345}
\end{figure}

Fig.~\ref{samplerate} shows the success rate for different sampling rates. In this experiment, there are four persons and the window size is set to 30 s. From Fig.~\ref{samplerate}, we can see that with the increase of the number of sampling rates, the success rate is also increased. It is noticed that the success rates for 5 Hz and 30 Hz are approximately 70\% and 90\%, respectively. As the sampling rate is increased, the length of the data for CP decomposition is increased for the 30 s window size case, which helps to improve the estimation accuracy. Furthermore, we find that the performance becomes stable when the sampling rate exceeds 20 Hz, indicating that a sampling rate of 20 Hz is sufficient for CP decomposition. Thus, we set the sampling rate to 20 Hz for the TensorBeat experiments. 

\begin{figure} 
\centerline{\includegraphics[width=3.5in]{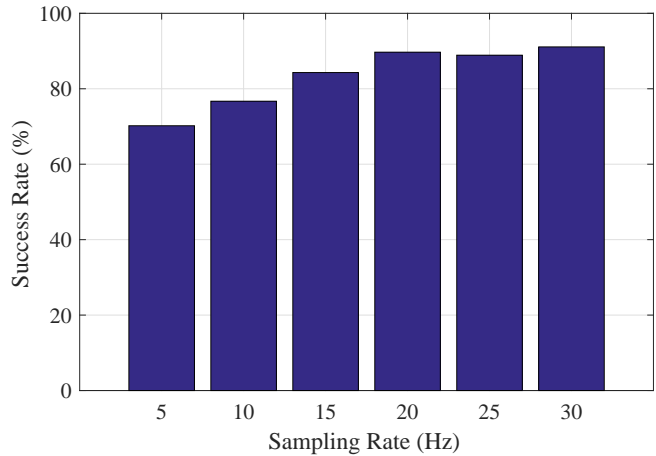}}
\caption{Success rates for different sampling rates (computer laboratory).}
\label{samplerate}
\end{figure}

Fig.~\ref{window_size} plots the success rates for different window sizes. This experiment is for the computer laboratory scenario with four persons and the sampling rate is set to 20 Hz. From Fig.~\ref{window_size}, we can see that the success rate is greatly increased by increasing the window size of the Hankel matrix from 15 s to 30 s. This is because Hankelizaiton will take half of the data to smooth the phase difference signal, which reduces the resolution in the time domain. Thus, we need to increase the window size to improve the estimation accuracy. Furthermore, the change of success rate is small for window sizes from 30 s to 45 s. Thus, we select the window size of 30 s for the TensorBeat experiments.

\begin{figure} 
\centerline{\includegraphics[width=3.5in]{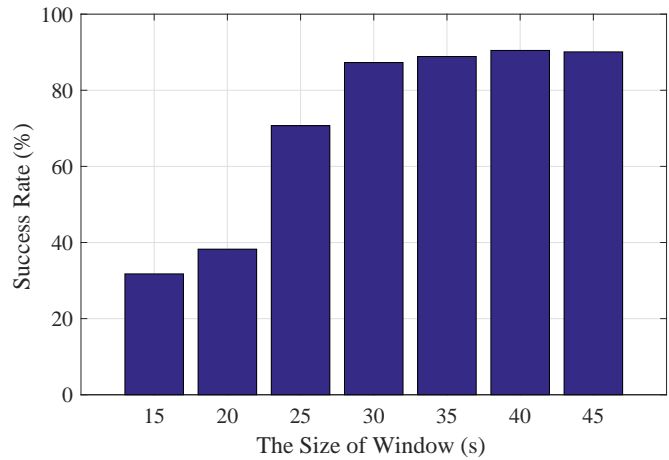}}
\caption{Success rates for different window sizes (computer laboratory).}
\label{window_size}
\end{figure}

Finally we examine the impact of LOS and NLOS scenarios. The success rates are plotted in Fig.~\ref{los_nlos}. 
In this experiment, we consider the challenge condition of the NLOS scenario, where all the persons stay on the LOS path between the transmitter and receiver, i.e., they form a straight line and block each other. 
From Fig.~\ref{samplerate}, we find that the performances for LOS and NLOS are nearly the same for the cases of two or three persons, where high estimation accuracy can be achieved. This is due to the WiFi multipath effect, which is regarded harmful in general but becomes helpful in breathing rate estimation when tensor decomposition is used. The breathing signal of every person can still be captured at the receiver from the phase difference data. However, when the number of the persons is further increased, the success rate will decrease quickly. In fact, the strength of the breathing signals for some persons will become too weak to be detected when there are too many people blocking each other. 

\begin{figure} 
\centerline{\includegraphics[width=3.5in]{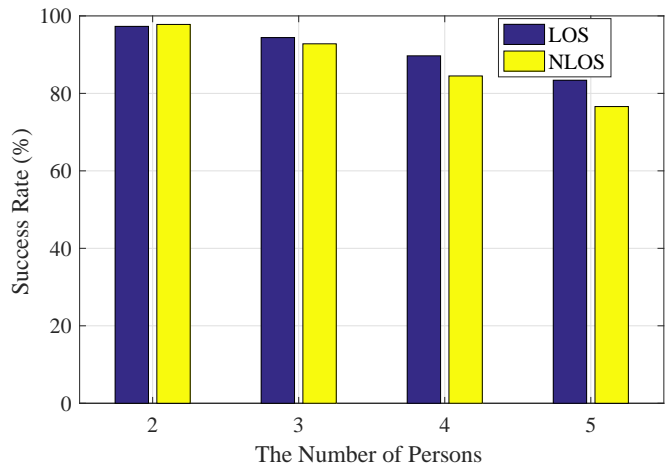}}
\caption{Success rates for (i) when multiple persons form a line in the LOS path between the transmitter and receiver; (ii) when multiple persons are in scattered around (computer laboratory).}
\label{los_nlos}
\end{figure}
\section{Related Work} \label{sec:related}
This work is closely related to sensor based and RF signal based vital signs monitoring, as well as CSI based indoor localization and human activity recognition, which are discussed in the following.

Sensor-based vital signs monitoring with wearable and smart devices employs the special hardware attached the person body to monitor breathing and heart rate data. The capnometer can measure carbon dioxide (CO2) concentrations in respired gases, which are employed to monitor patients breathing rate in hospital. However, it is uncomforted for patients to wear them, which are thus leveraged in clinical environments~\cite{Capnometers}. Photoplethysmography (PPG) is an optical technique to monitor the blood volume changes in the tissues by measuring light absorption changes, thus requiring the sensors attached to persons finger such as pulse oximeters~\cite{photoplethysmographic}. In addition, smartphone can use the camera to detect light changes from the video frames which can extract the PPG signal to monitor the heart rate~\cite{optical}. Moreover, the smartphones can also estimate breathing rate by using the built-in accelerometer, gyroscope~\cite{Zephyr} and microphone~\cite{Sleep}, which require persons to put smartphones near-by and wear sensors for breathing monitoring. However, these techniques based on sensors cannot be applied for remote monitoring vital signs.

RF based systems for vital signs monitoring use wireless signals to track the breathing-induced chest change of a person, which are mainly based on radar and WiFi techniques. For radar based vital signals monitoring, Vital-Radio employs frequency modulated continuous wave (FMCW) radar to estimate breathing and heart rates, even for two person subjects in parallel~\cite{smart}. But the system requires a custom hardware with a large bandwidth from 5.46 GHz to 7.25 GHz. For WiFi based vital signs monitoring, UbiBreathe system employ WiFi RSS for breathing rate monitoring, which, however, requires the device placed in the line of sight path between the transmitter and the receiver for estimating the breathing rate~\cite{UbiBreathe}. Moreover mmVital based on RSS can use 60 GHz millimeter wave (mmWave) signal for breathing and heart rates monitoring with the larger bandwidth about 7GHz, which cannot monitor the longer distance and require high gain directional antennas for the transmitter and the receiver~\cite{Millimeter}\cite{Gong10}. Recently, the authors leverage the amplitudes of CSI data to monitoring vital signs~\cite{Tracking}. This work is mainly to track the vital signs when a person is sleeping, which is limited for monitoring a maximum of two persons at the same time.

In additional to vital signs monitoring, recently, CSI based sensing systems have also been used for indoor localization and human activity recognition~\cite{SDR}. CSI-based fingerprinting systems have been proposed to obtain high localization accuracy. FIFS is the first work to uses the weighted average of CSI amplitude values over multiple antennas for indoor localization~\cite{FIFS}. To exploit the diversity among the multiple antennas and subcarriers, DeepFi leverage 90 CSI amplitude data from the three antennas with a deep autoencoder network for indoor localization~\cite{DeepFi}. Also, PhaseFi leverages calibrated CSI phase data for indoor localization based on deep learning~\cite{PhaseFi}. Different from CSI-based fingerprinting techniques, SpotFi system leverages a super-resolution algorithm to estimate the angle of arrival (AoA) of multipath components for indoor localization based on CSI data from three antennas~\cite{SpotFi}. On the other hand, E-eyes system leverages CSI amplitude values for recognizing household activities such as washing dishes and taking a shower~\cite{activity}. WiHear system employs specialized directional antennas to measure CSI changes from lip movement for determining spoken words~\cite{Hear}. CARM system considers a CSI based speed model and a CSI based activity model to build the correlation between CSI data dynamics and a given human activity~\cite{Wei}. Although CSI based sensing are effective for indoor localization and activity recognitions, there are few works for using CSI phase difference data to detect multiple persons behaviors at the same time.

The TensorBeat system is motivated by these interesting prior works. To the best of our knowledge, we are the first to leverage CSI phase difference data for multiple persons breathing rate estimation. We are also the first to employ tensor decomposition for RF sensing based vital signs monitoring, which can be also employed for indoor localization and human activity recognition.

\section{Conclusions} \label{sec:conC}

In this paper, we proposed TensorBeat, tensor decomposition for estimating multiple persons breathing beats with commodity WiFi. The proposed TensorBeat system employed CSI phase difference data to obtain the periodic signals from the movements of multiple breathing chests by leveraging tensor decomposition. We implemented several signal processing methods including data preprocessing, CP decomposition, signal matching algorithm, and peak detection in TensorBeat. We validate the performance of TensorBeat with extensive experiments under three indoor environments. Our analysis and experimental study demonstrated that the proposed TensorBeat system can achieve satisfactory performance for multiple persons breathing estimation. 

\bibliographystyle{ACM-Reference-Format-Journals}
\bibliography{acmsmall-sample-bibfile}

\end{document}